\documentclass[journal]{IEEEtran}

\IEEEoverridecommandlockouts
\usepackage{amsthm}
\usepackage{wrapfig}
\usepackage{amsthm}
\usepackage{mathrsfs}
\usepackage{graphicx}
\usepackage{graphics}
\usepackage{amssymb}
\usepackage{amsmath,mathtools}
\usepackage{color}
\usepackage{xspace}
\usepackage{algpseudocode}
\usepackage{bbm}
\usepackage{comment}
\usepackage{algorithmicx}
\usepackage{psfrag}
\usepackage{url}
\usepackage{float}
\usepackage{subfig}
\usepackage{soul,cite}

\usepackage{rotating}
\usepackage{chngcntr}
\usepackage{apptools}
\usepackage{graphicx}
\usepackage{graphics}
\usepackage{amssymb}
\usepackage{amsmath}
\usepackage{mathtools}
\usepackage{color}
\usepackage{xspace}
\usepackage{algpseudocode}
\usepackage{bbm}
\usepackage{comment}
\usepackage{algorithmicx}
\usepackage{psfrag}

\usepackage{amsmath}
\allowdisplaybreaks[4]
\usepackage{tikz}
\usetikzlibrary{shapes,arrows}
\usetikzlibrary{positioning}

\tikzstyle{block}=[draw opacity=0.7,line width=1.4cm]
\DeclareMathAlphabet{\mathpzc}{OT1}{pzc}{m}{it}
\definecolor{CranJ}{cmyk}{0,0.69,0.54,0.04} 
\definecolor{PinkJ}{cmyk}{0,0.71,0.43,0.12} 
\definecolor{Cran}{cmyk}{0,0.73,0.41,0.29} 
\definecolor{VRed}{cmyk}{0,0.75,0.25,0.2} 
\definecolor{ORed}{cmyk}{0,0.75,0.75,0} 
\definecolor{CBlue}{cmyk}{1,0.25,0,0} 

\newcommand{\Hf}{\operatorname{H}} 
\newcommand{\Df}{\operatorname{\vect{\delta}}} 
\newcommand{\ravg}{\mathsf{avg}^{\textup{a}}} 

\newcommand{\dravg}{\dot{\tilde{\mathsf{avg}}}^{\textup{a}}} 

\newcommand{\Dravg}{\Delta\vect{\mathsf{avg}}^{\textup{a}}_k} 
\newcommand{\ee}{\operatorname{e}}

\newcommand{\VV}{\mathcal{V}}
\newcommand{\EE}{\mathcal{E}}
\newcommand{\GG}{\mathcal{G}}

\newcommand{\lL}{\vectsf{L}}
\newcommand{\rR}{\vect{\mathfrak{R}}}
\newcommand{\rr}{\vect{\mathfrak{r}}}

\newcommand{\real}{{\mathbb{R}}}
\newcommand{\reals}{{\mathbb{R}}}
\newcommand{\integer}{{\mathbb{Z}}}

\newcommand{\integernonnegative}{{\mathbb{Z}_{\geq0}}}
 
\newcommand{\realpositive}{{\mathbb{R}}_{>0}}
\newcommand{\realnonnegative}{{\mathbb{R}}_{\ge 0}}

\newcommand{\Co}{\operatorname{Co}}

\newcommand{\until}[1]{\in\{1,\dots,#1\}}

\newcommand{\vect}[1]{\boldsymbol{\mathbf{#1}}}
\newcommand{\vectsf}[1]{\vect{\mathsf{#1}}}

\newcommand{\dvect}[1]{\dot{\vect{#1}}}

\newcommand{\Diag}[1]{\operatorname{Diag}(#1)}

\newcommand{\boxend}{\hfill \ensuremath{\Box}}

\newcommand{\yifan}[1]{{\color{blue}#1}}

\newtheorem{thm}{Theorem}

\newtheorem{lem}[thm]{Lemma}

\newtheorem{prob}{Problem}
\newtheorem{assump}{Assumption}

\newcommand{\oprocendsymbol}{\hbox{$\bullet$}}
\newcommand{\oprocend}{\relax\ifmmode\else\unskip\hfill\fi\oprocendsymbol}

\title{Dynamic Active Average Consensus and its Application in Containment Control
} 

\author{Yi-Fan Chung and Solmaz S. Kia, \emph{Senior Member, IEEE} 
\thanks{The authors are with the Department of Mechanical and Aerospace Engineering, University of California, Irvine, Irvine, CA 92697, {\tt \{yfchung,solmaz\}@uci.edu}. 
 This work is supported by NSF award IIS-SAS-1724331. 
}}

\begin{document}
\maketitle
\begin{abstract}
This paper proposes a continuous-time dynamic active weighted average consensus algorithm in which the agents can alternate between active and passive modes depending on their ability to access to their reference input. The objective is to enable all the agents, both active and passive, to track the weighted average of the reference inputs of the active agents. The algorithm is modeled as a switched linear system whose convergence properties are carefully studied considering the agents' piece-wise constant access to the reference signals and possible piece-wise constant weights of the agents. We also study the discrete-time implementation of this algorithm. Next, we show how a containment control problem, in which a group of followers should track the convex hull of a set of observed leaders, can be cast as an active average consensus problem, and solved efficiently by our proposed dynamic active average consensus algorithm. Numerical examples demonstrate our results.
\end{abstract}



\section{Introduction}
We propose a distributed solution for the dynamic active weighted average consensus problem and study its use in solving a distributed containment control problem. In dynamic active weighted average consensus problem, at any time, only a subset of the agents are active, meaning that only a subset of agents collects measurements. The objective then is to enable all the agents, both active  and passive, to obtain the weighted average of the collected measurements without knowing the set of active agents. The well-known average consensus problem, extensively studied in the literature for both static~\cite{ROS-JAF-RMM:07}
and dynamic~\cite{SSK-BVS-JC-RAF-KML-SM:19} 
reference signals, is in fact a special case of this problem with all the agents being active at all times and employing an equal weight of one.

The active weighted average consensus problem can be viewed as a \emph{weighted average consensus problem}~\cite{ROS-RMM:04}, in which the weights are $0$ for passive agents. However, the solutions for weighted average consensus (see e.g., \cite{ROS-RMM:04,CJ::06,SY::16}) use  the notation of the `equivalent' Laplacian matrix, which is the multiplication of the inverse of the weight matrix and the Laplacian matrix. 
Therefore, the weights should be non-zero, and thus these solutions cannot solve the active average consensus problem. Solutions specifically addressing the active (weighted) average consensus problem are proposed in~\cite{TY-JDP:14,JDP-TY-GC-SK:15,JDP-TY-JS-ELP:20}, but, they require both the reference inputs and their derivatives to be bounded to guarantee bounded error tracking. \cite{TY-JDP:14,JDP-TY-GC-SK:15}~also assume that the active and passive role of the agents are fixed and agents cannot alternative between modes. On the other hand,~\cite{JDP-TY-JS-ELP:20} allows the agents to change mode but requires that this change to be smooth.

In this paper, we propose a continuous-time solution for dynamic active weighted average consensus over connected graphs that requires only the rate of the change of the reference inputs to be bounded. Also, the agents can switch between active and passive modes or switch their weights instantaneously, as long as a dwell time exists between the switching incidences. Abrupt switching is usually the case for practical problems where agents are observing dynamic activities that can enter or leave the observation zone of the agents and thus change the agents' role from active to passive or vice versa in a non-smooth fashion. We model our algorithm as a switched linear system and study its convergence properties carefully by taking into account the piece-wise constant weights and access of the agents to the reference signals. Our study employs the concept of distributional derivatives~\cite{KRP::1998} to model the derivative of piece-wise continuous functions and characterize the transient error at the switching times.  

\begin{figure}
	\centering
	\includegraphics[trim= 0 0pt 0 0pt,clip,width=0.4\textwidth]{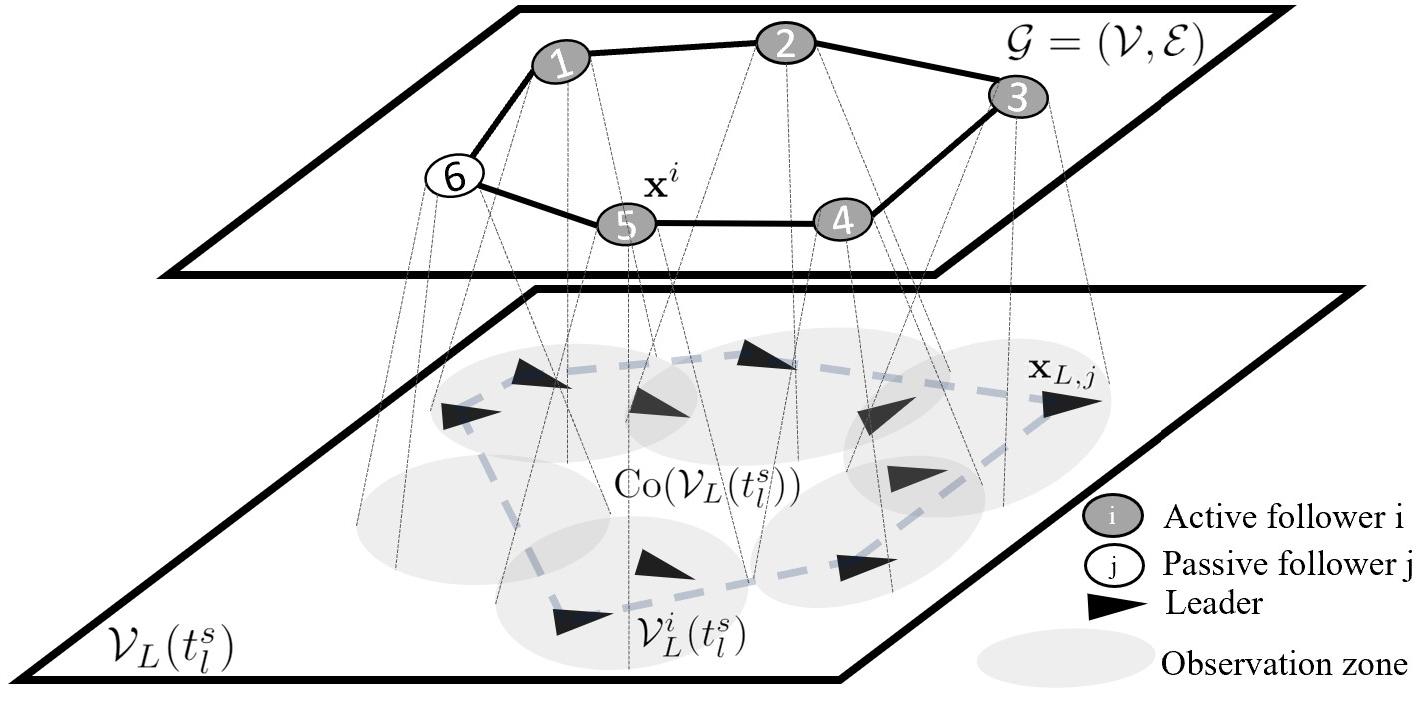}
	\caption{{\small A containment control scenario where a set of six followers should track the convex hull of a set of the dynamic leaders that they observe:  
	Followers $\VV_{\text{a}}=\{1,\cdots,5\}$ are active agents that each observes a subset of the leaders, while follower $6$ is the passive agent that should still follow the convex hull of the leaders despite having no measurement. Lemma~\ref{lem::cov_noverlap} below shows that the average of the geometric centers of the observed leaders at each active agent is a point in the convex hull of the leaders. Thus, this containment problem can be formulated as an active average consensus~problem.}}\vspace{-0.1in}
	\label{fig::containment_scenario}
\end{figure}
Our next contribution in this paper is studying the discrete-time implementation of our proposed dynamic active weighted average consensus algorithm and using it to solve a containment control problem where a group of followers should track the convex hull of a set of leaders that they observe. We show that the average of the geometric centers  of the observed leaders at each active agent is a point in the convex hull of the leaders. Thus, the containment problem can be formulated and solved as an active average consensus~problem, see Fig.~\ref{fig::containment_scenario}. 
Continuous-time solutions for containment problems can be found in~\cite{JM-FT-GE-EM-BA:08,WX-LS-SP:14,HL-GX-LW:12}. But, the requirement for continuous inter-agent information sharing can be of concern for practical problems where agents communication bandwidth is limited. Discrete-time containment control solutions where agents communicate with each other in a finite rate are given in~\cite{LG-GFT-RC:13,KZ-SJM-DW:16}. To provide perfect tracking,~\cite{JM-FT-GE-EM-BA:08,WX-LS-SP:14,HL-GX-LW:12,LG-GFT-RC:13,KZ-SJM-DW:16} assume that the leaders are static or if they are dynamic they either follow a certain dynamics that is known to the followers or the leaders' motions have to be coordinated with the followers. 
In this paper, considering tracking problems where the state of the leaders is only measured online,
we make no assumption about the dynamics of the leaders except that the change of the states of the leaders is bounded. This relaxation however, as known in dynamic consensus literature, is attained by trading off perfect tracking, as the online time-varying information takes some time to propagate through the network~\cite{SSK-BVS-JC-RAF-KML-SM:19}. A~preliminary version of our work appeared in~\cite{CY-KS:19}. There, we used two parallel conventional dynamic average consensus algorithms, one to generate the sum of the measurements divided by the size of the network and the other to obtain the sum of the active agents divided by the size of the network. Then, the average of the active measurements is obtained from dividing the output of the first algorithm by that of the second one. The current manuscript is offering a computationally more efficient algorithm, which has a lower communication complexity and avoids zero-crossing problem observed in our initial work~\cite{CY-KS:19} for its approach to solve dynamic active average consensus~problem.

\section{Notations and Preliminaries}\label{sec::notaion_prelim_relay}
We let $\reals$, $\realpositive$, $\realnonnegative$, $\mathbb{Z}$, $\mathbb{Z}_{> 0}$ and $\mathbb{Z}_{\geq 0}$
denote the set of real, positive real, non-negative real, integer, positive integer, and non-negative integer, respectively. For $\vect{s}\in\reals^d$,
$\|\vect{s}\|=\sqrt{\vect{s}^\top\vect{s}}$ denotes the standard
Euclidean norm. We let
$\vect{1}_n$ (resp. $\vect{0}_{n}$) denote the vector of $n$ ones
(resp. $n$ zeros), and $\vect{I}_n$ denote the $n\times n$ identity
matrix. When clear from the context, we do not
specify the matrix dimensions. 
$\Hf(t)=\begin{cases}0,& t<0 \\ 1,& t \geq 0 \end{cases}$ is the \emph{Heaviside step} function.  $\Df(t)=\begin{cases}\infty,& t=0 \\ 0,& t \neq 0 \end{cases}$ such that $\int_{-\infty}^{\infty}\Df(t)\text{d}t=1$ is the \emph{Dirac Delta} function. In~a network of $N$ agents,
 the aggregate vector of local variables $p^i\in\real$, $i\until{N}$, is denoted by $\vect{p} = ({p}^1,\dots,{p}^N)^\top\in\reals^N$. 
 
 Consider the piece-wise continuous~function
\begin{align}\label{eq::pwc_v}
   \vect{v}(t)=\left\{\begin{smallmatrix}\vect{v}_0(t), \quad&t_0\leq t < t_1,\\
    \vect{v}_1(t), \quad&t_1\leq t < t_2,\\
    \vdots\\
    \vect{v}_{\bar{k}}(t), \quad&t_{\bar{k}}\leq t
    \end{smallmatrix}\right.,\quad
\end{align}
where $\vect{v}_i \in \mathcal{C}^1,i\in\{1,\cdots,\bar{k}\}$. Using the Heaviside step function,~\eqref{eq::pwc_v} reads as $\vect{v}(t)=\vect{v}_0+\sum_{k=1}^{\bar{k}}(\vect{v}_k-\vect{v}_{k-1}) \Hf(t-t_{k})$. 
Then, following~\cite{KRP::1998}, the distributional derivative of $\vect{v}(t)$ is
\begin{align}\label{eq::v_derivative}
    \frac{\text{d}}{\text{d}t}\vect{v}&=\vect{\dot{\tilde{v}}}+\sum\nolimits_{k=1}^{\bar{k}}(\vect{v}(t_k^+)-\vect{v}(t_k^-))
    \Df(t\!-\!t_{k}), 
\end{align}
where $\vect{\dot{\tilde{v}}}\!=\!\dvect{v}_0\!+\!\sum_{k=1}^{\bar{k}}(\dvect{v}_k-\dvect{v}_{k-1})\Hf(t-t_{k})$ or equivalently
\begin{align*}
    \vect{\dot{\tilde{v}}}=\left\{\begin{smallmatrix}\dvect{v}_0(t), \quad&t_0\leq t < t_1,\\
    \dvect{v}_1(t), \quad&t_1\leq t < t_2,\\
    \vdots\\
    \dvect{v}_{\bar{k}}(t), \quad&t_{\bar{k}}\leq t.
    \end{smallmatrix}\right. 
    \end{align*}
We assume that the piece-wise continuous signals are right-continuous, i.e. $\vect{v}(t_k)=\vect{v}(t_k^+)$. 
Hereafter, we use the notation ` $\vect{\dot{\tilde{\empty}}}$\,' to represent $ \vect{\dot{\tilde{v}}}(t)=\begin{cases}\dvect{v}(t)&t\neq t_k\\
    \dvect{v}(t_k^+)&t= t_k\end{cases}$.



An undirected \emph{graph} is a triplet $\GG = (\VV ,\EE,
\vect{\sf{A}})$,~where $\VV=\{1,\dots,N\}$ is the \emph{node set} and
$\EE \subseteq \VV\times \VV$ is the \emph{edge set}, and $\vect{\sf{A}}\in\real^{N\times N}$ is a \emph{adjacency}
matrix such that $ \sf{a}_{ij} =\sf{a}_{ji}>0$ if $(i, j) \in\EE$ and $
\mathsf{a}_{ij} = 0$, otherwise.  
An edge $(i, j)$ from $i$ to $j$ means that agents $i$ and $j$ can communicate. 
A \emph{connected graph} is an undirected graph in which for
every pair of nodes there is a path connecting them. The
\emph{degree} of a node $i$ is $\mathsf{d}^i =\Sigma^N_{j =1} \mathsf{a}_{ij}$. 
The \emph{Laplacian} matrix is $\lL =
\vect{\mathsf{D}} - \vect{\mathsf{A}}$, where
$\vect{\mathsf{D}} =
\Diag{\mathsf{d}^1,\cdots, \mathsf{d}^N} \in
\reals^{N \times N}$. For connected graphs, $\lL\vect{1}_N=\vect{0}$ and $\vect{1}_N^T\lL=\vect{0}$. Moreover,  $\vectsf{L}$ has one eigenvalue $\lambda_1\!=\!0$, and the rest of the eigenvalues $\{\lambda_i\}_{i=2}^N$ are positive. $\vect{T}=[\rr\quad\rR]\in\real^{N\times N}$ is an orthonormal matrix, where $\rr=\frac{1}{\sqrt{N}}\vect{1}_N$ and $\rR\!\in\!\real^{N\times(N-1)}$ is any matrix that makes $\vect{T}^\top\vect{T}=\vect{T}\vect{T}^\top=\vect{I}$. For a connected graph,
    $\vect{T}^\top\lL\vect{T}\!=\!\left[\begin{smallmatrix}
    0&\vect{0}\\
    \vect{0}&~\lL^{+}
    \end{smallmatrix}\right]$, where $\lL^{+}=\rR^\top\lL\rR$,
 is a positive definite matrix with eigenvalues $\{\lambda_i\}_{i=2}^N\subset\real_{>0}$.

\begin{lem}\label{lem::B+L Hurwitz}
Suppose the nonzero matrix $\vect{E}\in\real^{N\times N}$ is a diagonal matrix whose diagonal elements are either $0$ or of positive real numbers, and $\lL$ is the Laplacian matrix of a connected~graph.
Then, $-(\vect{E}+\lL)$ is Hurwitz.
\end{lem}
\begin{proof}
Consider the system $\dvect{x}=-(\vect{E}+\lL)\vect{x}$. Now consider Lyapnov function $V=\frac{1}{2}\vect{x}^\top \vect{x}$. Then, $\dot{V}=-\vect{x}^\top \vect{E}\vect{x}-\vect{x}^\top\lL\vect{x}\leq0,$
because $-\vect{x}^\top \vect{E}\vect{x}\leq 0$ and $-\vect{x}^\top\lL\vect{x}\leq 0$. However, $\dot{V}\equiv 0$ happens when $-\vect{x}^\top \vect{E}\vect{x}=0$ and $-\vect{x}^\top\lL\vect{x}=0$. But, since $-\vect{x}^\top\lL\vect{x}=0$ if and only if $\vect{x}=\alpha \vect{1}$, $\alpha\in\real$ then $\dot{V}\equiv 0$ if $\vect{x}\equiv \vect{0}$.
Therefore, invoking~\cite[Theorem 4.11]{khalil:02}, we conclude that the system $\dvect{x}=-(\vect{E}+\lL)\vect{x}$ is uniformly exponentially stable.  Thus, $-(\vect{E}+\lL)$ is Hurwitz.
\end{proof}
\section{Problem Definition}\label{sec::problem_relay}
Consider a network of $N$  single integrator agents $\dot{x}^i=u^i$, $i\in\VV$,
 interacting over a connected undirected graph $\GG$. Suppose each agent $i\in\VV$ has access to a measurable locally essentially bounded reference signal  $\mathsf{r}^i:\realnonnegative\to\real$ in a possibly intermittent fashion. For every agent $i\in\VV$, we let $\eta^i(t)$ be the mode and weight indicator function for the agent $i\in\VV$, which is in $\realpositive$ if agent $i$ is active and has access to $\mathsf{r}^i(t)$ at time $t\in\real_{\geq0}$, and $0$ otherwise.
%
Let $\VV_{\textup{a}}(t)\subset\VV$ be the set of active agents at time $t\in\real_{\geq0}$, i.e., $\VV_{\textup{a}}(t)=\big\{i\in\VV\,\big|\,\eta^i(t)>0\big\}$. In what follows, we assume that $\eta^i(t)$ and $|\VV_{\textup{a}}(t)|$ are piece-wise constant functions of time, and $\VV_{\textup{a}}(t)\not=\emptyset$ for all  $t\in\real_{\geq0}$. We refer to an agent in
$\VV\backslash\VV_\text{a}(t)$ 
as the passive agent at time $t$.
\begin{prob}[Active weighted average consensus problem]\label{prob::continu}{\rm
The active average consensus problem over $\GG$ is defined as designing a distributed control input $u^i$ such that the agreement state $x^i(t)\in\real$ of every agent $i\in\VV$ tracks 
\begin{equation*}\quad\quad\mathsf{avg}^{\textup{a}}(t)=\frac{\sum_{i=1}^N{\eta^i(t)\,\mathsf{r}^i}(t) }{\sum_{i=1}^N \eta^i(t)}.
\quad\quad\quad\boxend
\end{equation*}} 
\end{prob}
In what follows, we first propose a distributed continuous-time algorithm to solve Problem~\ref{prob::continu}. Then, we present a  discrete-time implementation of this active weighted average consensus algorithm in which the agents sample the reference inputs with a rate of $1/\delta_s$ in a zero-order fashion. Lastly, we show how a containment problem can be cast as dynamic active (homogeneously weighted) average consensus problem and solved using our proposed algorithm. 

\section{Continuous-Time Dynamic Active Average Consensus}\label{sec::CT_relay consensus algorithm}
Our solution to solve Problem~\ref{prob::continu} over a connected undirected graph $\GG$ is 
\begin{subequations}\label{eq:CT_IF_algorithm}
\begin{align}
 \dot{x}^i(t)=\,&-\eta^i(t)(x^i(t)-\mathsf{r}^i(t))-\sum\nolimits_{j=1}^N\mathsf{a}_{ij}(x^i(t)-x^j(t))\nonumber\\
 &-\sum\nolimits_{j=1}^N\mathsf{a}_{ij}(v^i(t)-v^j(t))+\eta^i(t)\dot{\mathsf{r}}^i(t),\\
 \dot{v}^i(t)&=\sum\nolimits_{j=1}^N\mathsf{a}_{ij}(x^i(t)-x^j(t)),
\end{align}
\end{subequations}
with $x^i(0),v^i(0)\in\reals$, $i\in\VV$. Here, $v^i(t)\in\real$ is an internal state that acts as an integral action. Next, we study the convergence properties of~\eqref{eq:CT_IF_algorithm} by modeling it as a switched system and analyzing the collective response of the agents.  In what follows, we let  $\vect{E}(t)=\Diag{\eta^1(t),\cdots,\eta^N(t)}$.
$\vect{E}(t)$ can be considered as switching in the class of non-zero diagonal matrices $\{\vect{E}_p\}_{p\in\mathcal{P}}$, $\mathcal{P}$ is the index set, each of which has diagonal elements being either positive real or $0$.
That is $\vect{E}(t)=\vect{E}_{\sigma(t)}\not=\vect{0}$ with the switching signal  $\sigma(t):\realnonnegative\to\mathcal{P}$. We let $N_\sigma(0,t)$ denote the number of switchings of $\sigma(t)$ on the interval $[0,t)$. In our problem of interest, the following common assumption for switch linear systems holds \cite{ZL-GH:10,HJP-MAS:19}.
\begin{assump}\label{assump::dwell time exists}\rm{
 There exist some $N_0\in\mathbb{Z}_{\geq 0}$ and $\tau_{\text{D}}\in\real_{> 0}$ such that,  $N_\sigma(0,t)\leq N_0+\frac{t}{\tau_{\text{D}}}$, $t\in\real_{>0}$, where $\tau_{\text{D}}$ is called the average dwell time and $N_0$ is the chatter bound.}
 \boxend
\end{assump}
We let $\vect{\ravg}=\ravg \vect{1}$, 
$\Dravg=\vect{\ravg}(t_k^+)-\vect{\ravg}(t_k^-)$, $\vect{w}(t)=\vect{E}_{\sigma(t)}(\vectsf{r}(t)-\vect{\ravg}(t))$, and $\Delta\vect{w}_k=\vect{w}(t_k^+)-\vect{w}(t_k^-)$, where $t_k$, $k\in\mathbb{Z}_{\geq0}$ is the $k$th switching time of the switching signal $\sigma(t)$. Throughout this paper we assume $t_0=0$. Lastly, given a time $t\in\real_{\geq0}$, $\bar{k}\in\integer_{\geq0}$ is the largest integer such that $t_{\bar{k}}\leq t$.

For convenience in the correctness analysis of algorithm~\eqref{eq:CT_IF_algorithm}, we use the change of variables $\overline{\vect{e}}=\vect{T}^\top(\vect{x}-\vect{\ravg})$, $\vect{q}=[q_1\,\,\,\,\vect{q}_{2:N}^\top]^\top=\vect{T}^\top(\lL\vect{v}-\vect{w})$ to write the equivalent compact form of ~\eqref{eq:CT_IF_algorithm} as 
\begin{subequations}
 \begin{align}
  \dot{q}_1=0,\\
\!\!\!   \left[\begin{smallmatrix}
    \dot{\overline{\vect{e}}}\\ \dvect{q}_{2:N}
    \end{smallmatrix}\right]=\,&\overline{\vect{A}}_{\sigma(t)}\! \left[\begin{smallmatrix}
    \overline{\vect{e}}\\ \vect{q}_{2:N}
    \end{smallmatrix}\right]+\overline{\vect{B}} \left[\begin{smallmatrix}\vect{E}\,\dot{\vectsf{r}}-\vect{\dravg}\\ -\vect{\dot{\tilde{w}}}\end{smallmatrix}\right]-\nonumber\\
    &\overline{\vect{B}}\sum\nolimits_{k=1}^{\bar{k}} \left[\begin{smallmatrix}\Dravg\\ \Delta\vect{w}_k\end{smallmatrix}\right]\Df(t-t_k).\label{eq::CT_IF_compact form}
\end{align}
\end{subequations}
where $\overline{\vect{A}}_{\sigma(t)}=\left[\begin{smallmatrix}
    -\vect{T}^\top(\vect{E}_{\sigma(t)}+\lL)\vect{T} &~~ -\left[\begin{smallmatrix}\vect{0}\\\vect{I}_{N-1} \end{smallmatrix}\right] \\
    \left[\begin{smallmatrix}\vect{0} & ~~\lL^{+}\lL^{+} \end{smallmatrix}\right]          & \vect{0}
    \end{smallmatrix}\right]$ and $\overline{\vect{B}}=\left[\begin{smallmatrix} 
    \vect{T}^\top & \vect{0}\\ \vect{0} &  \rR^\top
    \end{smallmatrix}\right]$.
Here, we used the facts that $\rr^\top\vect{\dot{\tilde{w}}}=0$ and $\rr^\top\Delta\vect{w}_k=0$. Also, we used $\rR\rR^\top\lL=\lL$ to write $\rR^\top\lL\lL\rR=\lL^{+}\lL^{+}$. Lastly, note that since $\vect{\ravg}$ and $\vect{w}$ are piece-wise continuous functions, we used~\eqref{eq::v_derivative} to compute their derivatives that appear in $ \dot{\overline{\vect{e}}}$ and $\dvect{q}$. Using standard results for linear time-varying systems we can write
    \begin{align}\label{eq::traj_compact}
  \left[\begin{smallmatrix}
    {\overline{\vect{e}}}(t)\\ \vect{q}_{2:N}(t)
    \end{smallmatrix}\right]\!=\,&\Phi(t,0)\!\left[\begin{smallmatrix}\overline{\vect{e}}(0)\\ \vect{q}_{2:N}(0)
    \end{smallmatrix}\right]\!\!+\!\!\int_0^t \!\!\!\Phi(t,\tau)\overline{\vect{B}} \Big(\! \left[\begin{smallmatrix}\vect{E}\,\dot{\vectsf{r}}-\vect{\dravg}\\ -\vect{\dot{\tilde{w}}}\end{smallmatrix}\right]\nonumber\\
    &-\sum\nolimits_{k=1}^{\bar{k}} \left[\begin{smallmatrix}\Dravg\\ \Delta\vect{w}_k\end{smallmatrix}\right]\Df(\tau-t_k) \Big)\text{d}\tau,
\end{align}
where $\Phi(t,\tau)$ is the transition matrix of linear system~\eqref{eq::CT_IF_compact form}. 
The next result shows that the internal dynamics of~\eqref{eq::CT_IF_compact form} is uniformly exponentially stable. Therefore, there always exists $\kappa_s,\lambda_s$ such that
\begin{align}\label{eq::phi}\left\|\Phi(t,\tau)\right\|\leq \kappa_s \textup{e}^{-\lambda_s (t-\tau)}, ~~ t\geq \tau\geq0.\end{align}
\begin{lem}\label{lem::CT_IF_algorithm_internal}
Let $\GG$ be a connected undirected graph. Then, every subsystem matrix $\overline{\vect{A}}_p$, $p\in\mathcal{P}$ of~\eqref{eq::CT_IF_compact form} is Hurwitz. Furthermore, under Assumption~\ref{assump::dwell time exists} the internal dynamics of~\eqref{eq::CT_IF_compact form} is uniformly exponentially stable, i.e.,~\eqref{eq::phi} holds.
\end{lem}
\begin{proof}
Consider the radially unbounded quadratic Lyapunov function $V=\frac{1}{2}\vect{q}_{2:N}^\top(\lL^{+}\lL^{+})^{-1}\vect{q}_{2:N}+\frac{1}{2}\overline{\vect{e}}^\top\overline{\vect{e}}$ (a common Lyapunov function for all the subsystems $\overline{\vect{A}}_{p\in\mathcal{P}}$ of the switched system $\overline{\vect{A}}_{\sigma(t)}$).
Here, note that since $\lL^{+}>0$, then $\lL^{+}\lL^{+}>0$. The Lie derivative of V along the trajectories of internal dynamics of~\eqref{eq::CT_IF_compact form} is
\begin{align}\label{eq::dotV}
  \dot{V}=&-\overline{\vect{e}}^\top\vect{T}^\top(\vect{E}_{p}+\lL)\vect{T}\,\overline{\vect{e}}\,\leq 0, \quad  {p\in\mathcal{P}}.
\end{align}
To establish negative semi-definiteness of $\dot{V}$, we invoke Lemma~\ref{lem::B+L Hurwitz}. So far we have established that $V$ is a weak Lyapunov function. Next, we use the LaSalle invariant principle and~\cite[Theorem 4]{HJP::04} to establish exponential stability of the internal dynamics of~\eqref{eq::CT_IF_compact form}. Let $\mathcal{S}_p=\{(\overline{\vect{e}},\vect{q}_{2:N})\in\real^{N}\times\real^{N-1}|\dot{V}\equiv0\}$ for all $p\in\mathcal{P}$. Given~\eqref{eq::dotV}, we then have $\mathcal{S}_p=\{(\overline{\vect{e}},\vect{q}_{2:N})\in\real^{N}\times\real^{N-1}|\overline{\vect{e}}=0\}$, for all $p\in\mathcal{P}$. Then, it is straightforward to observe that the trajectories of the internal dynamics of~\eqref{eq::CT_IF_compact form} that belong to $\mathcal{S}_{p\in\mathcal{P}}$, should also satisfy  $\vect{q}_{2:N}\equiv\vect{0}$. Therefore, the largest invariant set of the internal dynamics of~\eqref{eq::CT_IF_compact form} in $\mathcal{S}_{p\in\mathcal{P}}$ is the origin. Thus, using~\cite[Theorem 4.4]{khalil:02} all the subsystems $\overline{\vect{A}}_{p\in\mathcal{P}}$ of the switched system $\overline{\vect{A}}_{\sigma(t)}$ are globally asymptotically stable. Moreover, because the all subsystems of the switched system share the common weak quadratic Lyapunov function and the largest invariant set of $\mathcal{S}_{p\in\mathcal{P}}$ contains only the origin, given Assumption~\ref{assump::dwell time exists}, by virtue of~\cite[Theorem 4]{HJP::04} the internal dynamics of~\eqref{eq::CT_IF_compact form}, which is a switched system, is uniformly exponentially stable. Here, we note that according to~\cite[Theorem 2.1]{AZ:82} the origin being the largest invariant set of $\mathcal{S}_p$, for all $p\in\mathcal{P}$, ensures that the observability condition in~\cite[Theorem 4]{HJP::04} is satisfied.
\end{proof}
Given~\eqref{eq::traj_compact} and~\eqref{eq::phi}, we can characterize the tracking performance of  active average consensus algorithm~\eqref{eq:CT_IF_algorithm} as follows.
\begin{thm}\label{thm::CT_IF_algorithm}
Let $\GG$ be a connected undirected graph and suppose Assumption~\ref{assump::dwell time exists} holds.  Then, starting from any $x^i(0),v^i(0)\in\real$, $i\in\VV$ the trajectories of dynamic active average consensus algorithm~\eqref{eq:CT_IF_algorithm} satisfy
\begin{align}\label{eq:CT_IF_error bound}
    |x^i(t)-\ravg(t)&|\leq\kappa_s\ee^{-\lambda_s t}\left\|\left[\begin{smallmatrix}\vect{x}(0)-\vect{\ravg}(0)\\ \lL\vect{v}(0)-\vect{w}(0)\end{smallmatrix}\right]\right\|\nonumber\\
    &\!\!\!\!\!\!\!\!\!\!\!\!\!\!\!+\kappa_s\sum\nolimits_{k=1}^{\bar{k}}\ee^{-\lambda_s(t-t_k)}\left\|\left[\begin{smallmatrix}\Dravg\\\Delta\vect{w}_k \end{smallmatrix}\right]\right\|\Hf(t-t_k)\nonumber\\
    &\!\!\!\!\!\!\!\!\!\!\!\!\!\!\!+\frac{\kappa_s}{\lambda_s}\underset{0\leq\tau\leq t}{\sup}\left\|\left[\begin{smallmatrix}\vect{E}_{\sigma(\tau)}\,\dot{\vectsf{r}}(\tau)-\vect{\dravg}(\tau)\\ -\vect{\dot{\tilde{w}}}(\tau)\end{smallmatrix}\right]\right\|.
\end{align}
\end{thm}

\begin{proof}
We note that
 $\|\overline{\vect{B}}\|\leq 1$. Then, given~\eqref{eq::traj_compact} and~\eqref{eq::phi}, we can write 
 \begin{align*}
    \left\|\!\left[\begin{smallmatrix}\overline{\vect{e}}(t)\\\vect{q}_{2:N}(t)\end{smallmatrix}\right]\!\right\|
    &\!\leq\! \kappa_s\!\ee^{-\lambda_s t}\!\left\|\!\left[\begin{smallmatrix}\overline{\vect{e}}(0)\\\vect{q}_{2:N}(0)\end{smallmatrix}\right]\!\right\|\!+\!\kappa_s\!\! \int_0^t\!\!\! \ee^{-\lambda_s(t-\tau)}\!\left\|\!\left[\begin{smallmatrix}\vect{E}\,\dot{\vectsf{r}}-\vect{\dravg}\\ -\vect{\dot{\tilde{w}}}\end{smallmatrix}\right]\!\right\|\!\text{d}\tau\\
    &\!\!\!\!\!\!\!\!\!\!\!\!\!\!+\kappa_s \sum\nolimits_{k=1}^{\bar{k}} \int_0^t \ee^{-\lambda_s(t-\tau)}\left\| \left[\begin{smallmatrix}\Dravg\\ \Delta\vect{w}_k\end{smallmatrix}\right]\Df(\tau-t_k)\right\|\text{d}\tau.
\end{align*}
Then, the H\"older inequality is used to bound the second term of the right hand side to arrive at
\begin{align*}
   & \left\|\left[\begin{smallmatrix}\overline{\vect{e}}(t)\\\vect{q}_{2:N}(t)\end{smallmatrix}\right]\right\|
    \leq \kappa_s\ee^{-\lambda_s t}\left\|\left[\begin{smallmatrix}\overline{\vect{e}}(0)\\\vect{q}_{2:N}(0)\end{smallmatrix}\right]\right\|\\
     &\quad \qquad+\frac{\kappa_s}{\lambda_s} \sup_{0\leq\tau\leq t}\left\|\left[\begin{smallmatrix}\vect{E}_{\sigma(\tau)}\,\dot{\vectsf{r}}(\tau)-\vect{\dravg}(\tau)\\ -\vect{\dot{\tilde{w}}}(\tau)\end{smallmatrix}\right]\right\|\\
    &\quad+\kappa_s \sum\nolimits_{k=1}^{\bar{k}}\int_{0}^{t} \ee^{-\lambda_s(t-\tau)}\left\| \left[\begin{smallmatrix}\Dravg\\ \Delta\vect{w}_k\end{smallmatrix}\right]\Df(\tau-t_k)\right\|\text{d}\tau.
\end{align*}
Consequently, with integration by parts, the last term is equivalent to $\kappa_s\sum_{k=1}^{\bar{k}}\ee^{-\lambda_s(t-t_k)}\left\|\left[\begin{smallmatrix}\Dravg\\ \Delta\vect{w}_k \end{smallmatrix}\right]\right\| \Hf(t-t_k)$. Then, since $\vect{T}$ is an orthonormal matrix, we have $\left\|\left[\begin{smallmatrix}\overline{\vect{e}}(0)\\\vect{q}_{2:N}(0)\end{smallmatrix}\right]\right\|=\left\|\left[\begin{smallmatrix}\vect{x}(0)-\vect{\ravg}(0)\\ \lL\vect{v}(0)-\vect{w}(0)\end{smallmatrix}\right]\right\|$ and $\|\vect{x}-\vect{\ravg}\|=\|\overline{\vect{e}}\|$. Finally,  \eqref{eq:CT_IF_error bound} is derived along with the relation 
$|x^i-\ravg|\leq\left \|\left[\begin{smallmatrix}
    \overline{\vect{e}}^\top~&~ \vect{q}_{2:N}^\top
    \end{smallmatrix}\right]^\top\right\|$.
\end{proof}
We note that the first summand of the tracking error bound~\eqref{eq:CT_IF_error bound} is the transient response, which vanishes over time. The second summand is due to the agents alternating between active and passive sets or active agents switching their weights. If the average dwell time $\tau_D$ is large, this error also disappears after a while. The third summand can result in a steady-state error. This error that is expected in dynamic average consensus algorithms, as tracking an arbitrarily fast average signal with zero error is not feasible unless agents have some priori information about the dynamics generating the signals~\cite{SSK-BVS-JC-RAF-KML-SM:19}.  However, the size of this error is proportional to the rate of change of the signals and can be limited by limiting the rate.  We recall that to provide bounded tracking, previous work in~\cite{TY-JDP:14,JDP-TY-GC-SK:15,JDP-TY-JS-ELP:20} require both the reference input signals and their rate of change to be bounded. If the local reference signals are static and the agents do not switch, the agents exponentially converge to $\ravg$ without steady-state error.
Lastly, algorithm \eqref{eq:CT_IF_algorithm} does not require specific initialization. In other words, the convergence property of algorithm \eqref{eq:CT_IF_algorithm} uniformly holds for any initialization. Therefore, as long as the graph stays connected, agents can leave and join the network without effecting the convergence guarantees. 
Figure~\ref{fig::CT_Demo_leaving} demonstrates the performance of algorithm~\eqref{eq:CT_IF_algorithm} in a numerical example. 
\begin{figure}
	\centering
	\includegraphics[trim=0pt  0 0 0,clip,width=0.45\textwidth]{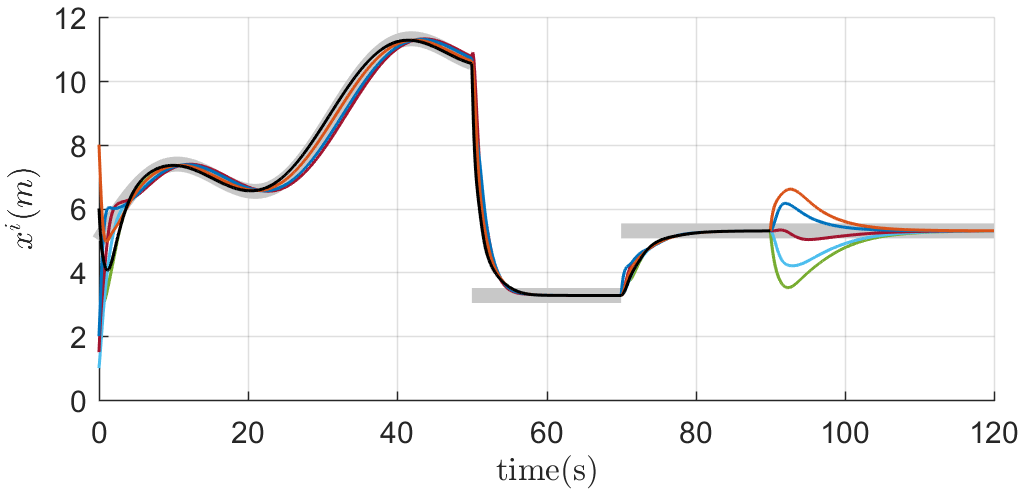}
	\caption{{\small A network of 6 agents with a ring interaction topology executes the active average consensus algorithm \eqref{eq:CT_IF_algorithm}. 
	In time interval $t\in[0,50)$, the observing agents  $\VV_{\textup{a}}(t)=\{1,2,4,6\}$ all have dynamic inputs. The observing agents at $t\in[50,70)$ and $t\in[70,120]$ are, respectively, $\VV_{\textup{a}}(t)=\{2,3,5,6\}$ $\VV_{\textup{a}}(t)=\{3,6\}$ and their observations are static signals. Agent 1 (black line) leaves the network at $t=90$. The gray thick line represents $\ravg(t)$. The agents can track the dynamic $\ravg(t)$ with bounded error in $t\in[0,50)$, while their tracking error is close to zero for the rest of the time as the reference signals are constant after $t=50$. The transient tracking error at time $t=70$ is due to switching of some of agents to the passive mode. This error is captured by the second term in the right-hand side of~\eqref{eq:CT_IF_error bound}. Lastly, agent 1's leaving causes perturbations at $t=90$ but the network still converge to $\ravg(t)$.}}
	\label{fig::CT_Demo_leaving}
\end{figure}

\section{Discrete-Time Dynamic Active Average Consensus}\label{sec::DT_relay consensus algorithm}
We consider a scenario where active agents sample their reference inputs at sampling times $t^s_l=l\delta_s \in\realnonnegative$, $l\in\mathbb{Z}_{\geq0}$, $\delta_s\in\real_{>0}$. The agents can communicate at discrete-times $t_k^c=k\delta_c \in\realnonnegative$, $k\in\mathbb{Z}_{\geq0}$, $\delta_c\in\real_{>0}$. The objective of every agent $i\in\VV$ is to track $\ravg(k)$ (where $k$ is the shorthand for~$t^c_k$). 
To solve the active average consensus problem under this scenario, we propose that every agent $i\in\VV$ implements 
\begin{subequations}\label{eq::DT_IF_algorithm}
\begin{align}
 x^i(k)&=z^i(k)+\eta^i(k) \mathsf{r}^i(k),\label{eq::DT_IF_algorithm_a}\\
 z^i(k+1)&=z^i(k)-\delta_c\eta^i(k)(x^i(k)-\mathsf{r}^i(k))\label{eq::DT_IF_algorithm_b}\\
 &\!\!\!\!\!\!\!\!\!\!\!\!\!\!\!\!\!\!\!\!-\delta_c\sum_{j=1}^N\mathsf{a}_{ij}(x^i(k)-x^j(k))-\delta_c \sum_{j=1}^N\mathsf{a}_{ij}(v^i(k)-v^j(k)),\nonumber\\
 v^i(k+1)&=v^i(k)+\delta_c\sum\nolimits_{j=1}^N\mathsf{a}_{ij}(x^i(k)-x^j(k)).\label{eq::DT_IF_algorithm_c}
\end{align}
\end{subequations}
which is an Euler discretized implementation of the active average algorithm~\eqref{eq:CT_IF_algorithm} with stepszie $\delta_c$.
Here, we assume that if $\delta_s\not=\delta_c$, the agents perform a zero-order hold sampling, so that $\mathsf{r}^i(k)=\mathsf{r}^i(\bar{l})$, $i\in\VV$, where $\bar{l}$ is the latest sampling time step such that $t^s_{\bar{l}}\leq t^c_k$. 
We let $\sigma(k):\integernonnegative\to\mathcal{P}$ be the switching signal of $\vect{E}(k)$, i.e., $\vect{E}(k)=\vect{E}_{\sigma(k)}$. 
Then, we implement the same change of variable as for the continuous-time algorithm~\eqref{eq:CT_IF_algorithm} to write the compact form of \eqref{eq::DT_IF_algorithm} as
\begin{subequations}\label{eq::DT_IF_algorithm_equvqlentform}
\begin{align}
q_1(k+1)&=q_1(k),\\
    \left[\begin{smallmatrix}
    \overline{\vect{e}}(k+1)\\ \vect{q}_{2:N}(k+1)
    \end{smallmatrix}\right]&\!=\!(\vect{I}+\delta_c\overline{\vect{A}}_\sigma) \!\left[\begin{smallmatrix}
    \overline{\vect{e}}(k)\\ \vect{q}_{2:N}(k)
    \end{smallmatrix}\right]\!+\!
        \vect{\overline{B}} \left[\begin{smallmatrix} \Delta\vect{E}\vectsf{r}(k)-\Delta\vect{\ravg}(k) \\ -\Delta\vect{w}(k)\end{smallmatrix}\right],\label{eq::DT_IF_compact form}
\end{align}
\end{subequations}
where $\overline{\vect{A}}_\sigma$ and $\vect{\overline{B}}$ are defined in \eqref{eq::CT_IF_compact form}, $\Delta\vect{E}\vectsf{r}(k)=\Delta\vect{E}(k+1)\vectsf{r}(k+1)-\Delta\vect{E}(k)\vectsf{r}(k)$, $\Delta\vect{\ravg}(k)=\vect{\ravg}(k+1)-\vect{\ravg}(k)$, and $\Delta\vect{w}(k)=\vect{w}(k+1)-\vect{w}(k)$. Then, given $|x^i-\ravg|\leq\left \|\left[\begin{smallmatrix}
    \overline{\vect{e}}^\top~&~ \vect{q}_{2:N}^\top
    \end{smallmatrix}\right]^\top\right\|$, the tracking performance of~\eqref{eq::DT_IF_algorithm} can be understood by studying the convergence properties of~\eqref{eq::DT_IF_compact form}.
    For the discrete-time implementation, the following assumption holds.
    \begin{assump}\label{assump::switching in finite set}
    The switched system \eqref{eq::DT_IF_compact form} switches in a finite set of subsystem, i.e., $\sigma(k):\integernonnegative\to\overline{\mathcal{P}}$, where $\overline{\mathcal{P}}\subset\mathcal{P}$ is a finite subset. 
    \end{assump} The first result below shows that with a proper choice for $\delta_c$ every subsystem $(\vect{I}+\delta_c\overline{\vect{A}}_p)$, $p\in\mathcal{P}$ is Schur.
    However, this is not enough to guarantee that the internal dynamics of~\eqref{eq::DT_IF_compact form} is exponentially stable. To provide such guarantee, following~\cite[Corollary 1]{ZG-HB-YK-MA:02}, we impose the following standard assumption.
\begin{assump}\label{assump::slow switching}{\rm
The average dwell time $\tau_D$ of the switching signal $\sigma(k)$ satisfies $\tau_D\!\geq\!\tau^*_D$, where $\tau^*_D$ is a~stable average dwell time of the switched~system~\eqref{eq::DT_IF_compact form}. }\boxend
\end{assump}
Note that $\tau_D^*$ of the switched system \eqref{eq::DT_IF_compact form} can be computed using the methods introduced in \cite{ZG-HB-YK-MA:02,GJC-CP::06}.

\begin{lem}\label{lem::DT_IF_algorithm_internal}
Let $\GG$ be a connected undirected graph. Then, every subsystem matrix $(\vect{I}+\delta_c\overline{\vect{A}}_p)$, $p\in\overline{\mathcal{P}}$ of~\eqref{eq::DT_IF_compact form} is Schur provided $\delta_c\in(0,\bar{d})$, where $$\bar{d}=\min \left\{\{-2\frac{\text{Re}(\mu_{i,p})}{|\mu_{i,p}|^2}\}_{i=1}^{2N-1}\right\}_{p\in\overline{\mathcal{P}}}$$ and  $\{\mu_{i,p}\}_{i=1}^{2N-1}$ are the set of eigenvalues of   $\overline{\vect{A}}_{p}$. Furthermore, under Assumption~\ref{assump::slow switching} the internal dynamics of~\eqref{eq::DT_IF_compact form} is uniformly exponentially stable, i.e., there always exists  $\kappa_d\in\realpositive$ and $\omega_d\in(0,1)$, such that, the state transition matrix $\Phi(k,j)$ of ~\eqref{eq::DT_IF_compact form} satisfies 
\begin{equation}\label{eq::exp_DT}\|\Phi(k,j)\|\leq\kappa_d\,\omega_d^{(k-j)},  \quad k\geq j\geq 0, k,j\in\integernonnegative 
.\end{equation}
\end{lem}
\begin{proof}
Lemma~\ref{lem::CT_IF_algorithm_internal} ensures that every $\overline{\vect{A}}_{p}$, $p\in\overline{\mathcal{P}}\subset\mathcal{P}$ is a Hurwitz matrix. Then, it follows from~\cite[Lemma S1]{SSK-BVS-JC-RAF-KML-SM:19} that $(\vect{I}+\delta_{c,p}\overline{\vect{A}}_p)$, $p\in\overline{\mathcal{P}}$ is Schur if $\delta_{c,p}\in(0,\bar{d}_p)$, where  $\bar{d}_p=\min\{-2\frac{\text{Re}(\mu_{i,p})}{|\mu_{i,p}|^2}\}_{i=1}^{2N-1}$. As a result, $(\vect{I}+\delta_c\overline{\vect{A}}_p)$, $p\in\overline{\mathcal{P}}$ is Schur if $\delta_c\in(0,\bar{d})$, where $\bar{d}=\min\{\bar{d}_p\}_{p\in\overline{\mathcal{P}}}$. Then, given Assumption~\ref{assump::slow switching}, it follows from~\cite[Corollary 1]{ZG-HB-YK-MA:02} that the zero input dynamics of switched system~\eqref{eq::DT_IF_compact form} is uniformly exponentially stable.  
\end{proof}

The next result characterizes the tracking performance of~\eqref{eq::DT_IF_algorithm}. 
 
\begin{thm}\label{thm::DT_IF_algorithm}
Let $\GG$ be a connected undirected graph and suppose Assumption~ \ref{assump::switching in finite set} and \ref{assump::slow switching} hold.. Then, for any $\delta_c\in(0,\bar{d})$, starting from any $x^i(0),v^i(0)\in\real$, $i\in\VV$, the trajectories of dynamic active average consensus algorithm~\eqref{eq::DT_IF_algorithm} satisfy
\begin{align}\label{eq::DT_IF_error bound}
    |x^i(k)&-\ravg(k)|\leq \kappa_d\omega_d^k\left\|\left[\begin{smallmatrix}\vect{x}(0)-\vect{\ravg}(0)\\ \lL\vect{v}(0)-\vect{w}(0)\end{smallmatrix}\right]\right\|+\nonumber\\
    &\frac{\kappa_d(1-\omega_d^k)}{1-\omega_d}\sup_{0\leq l \leq k-1}\left\|\left[\begin{smallmatrix} \Delta\vect{E}\vectsf{r}(l)-\Delta\vect{\ravg}(l) \\ -\Delta\vect{w}(l)\end{smallmatrix}\right]\right\|.
\end{align}
\end{thm}
\begin{proof} Using standard results for linear systems, trajectories of~\eqref{eq::DT_IF_compact form} are given by 
\begin{align*}
    \left[\begin{smallmatrix}
    \overline{\vect{e}}(k)\\ \vect{q}_{2:N}(k)
    \end{smallmatrix}\right]=&\Phi(k,0)\left[\begin{smallmatrix}
    \overline{\vect{e}}(0)\\ \vect{q}_{2:N}(0)
    \end{smallmatrix}\right]\\
    &+\sum\nolimits_{j=0}^{k-1}\Phi(k,j+1)\vect{\overline{B}} \left[\begin{smallmatrix} \Delta\vect{E}\vectsf{r}(j)-\Delta\ravg(j) \\ -\Delta\vect{w}(j)\end{smallmatrix}\right].
\end{align*}
Then, given that $\|\vect{\overline{B}}\|\leq 1$ and~\eqref{eq::exp_DT} we can write
\begin{align*}
        \left\|\left[\begin{smallmatrix}
    \overline{\vect{e}}(k)\\ \vect{q}_{2:N}(k)
    \end{smallmatrix}\right]\right\|&\leq\kappa_d\, \omega_d^k\left\|\left[\begin{smallmatrix}
    \overline{\vect{e}}(0)\\ \vect{q}_{2:N}(0)
    \end{smallmatrix}\right]\right\|\\
    &\!\!\!\!\!\!\!\!\!\!\!+\kappa_d \sum\nolimits_{j=0}^{k-1}\omega_d^j\sup_{0\leq l \leq k-1}\left\|\left[\begin{smallmatrix} \Delta\vect{E}\vectsf{r}(l)-\Delta\vect{\ravg}(l) \\ -\Delta\vect{w}(l)\end{smallmatrix}\right]\right\|.
\end{align*}
By the sum of geometric sequence, $\kappa_d \sum_{j=0}^{k-1}\omega_d^j=\frac{\kappa_d(1-\omega_d^k)}{1-\omega_d}$. Then, given that $\left\|\left[\begin{smallmatrix}\overline{\vect{e}}(0)\\\vect{q}_{2:N}(0)\end{smallmatrix}\right]\right\|=\left\|\left[\begin{smallmatrix}\vect{x}(0)-\vect{\ravg}(0)\\ \lL\vect{v}(0)-\vect{w}(0)\end{smallmatrix}\right]\right\|$ and $|x^i-\ravg|\leq\left \|\left[\begin{smallmatrix}
    \overline{\vect{e}}^\top~&~ \vect{q}_{2:N}^\top
    \end{smallmatrix}\right]^\top\right\|$, tracking error \eqref{eq::DT_IF_error bound} is established. 
\end{proof}

\section{Distributed containment control via dynamic active average consensus modeling}\label{sec::containment}
In this section, we use the discrete-time dynamic active weighted average consensus algorithm to solve a containment control problem. Consider a group of $M$ ($M$ can change with time) mobile leaders that are moving with a bounded velocity on a $\mathbb{R}^2$ or $\mathbb{R}^3$ space. $\vect{x}_{L,j}(t)$ represents the position vector of leader $j\until{M}$ at time $t\in\real_{\geq0}$. A set of networked follower agents $\VV=\{1,\cdots,N\}$ interacting over a connected graph $\GG$ monitors the leaders. The agents can communicate at discrete-times $t_k^c=k\delta_c \in\realnonnegative$, $k\in\mathbb{Z}_{\geq0}$, $\delta_c\in\real_{>0}$. The agents sample the leaders at sampling times $t^s_l=l\delta_s \in\realnonnegative$, $l\in\mathbb{Z}_{\geq0}$, $\delta_s\in\real_{>0}$. We let $\VV_L^i(t^s_l)$ be the set of leaders observed by agent $i\in\VV$ at sampling time $t^s_l$. Between each sampling time, agent $i\in\VV$ uses $\vect{x}_{L,j}(t)=\vect{x}_{L,j}(t^s_l)$ and $\VV_L^i(t)=\VV_L^i(t^s_l)$, $t\in[t^s_l,t^s_{l+1})$, $l\in\mathbb{Z}_{\geq0}$, $j\in\VV_L^i(t^s_l)$. At every sampling time $t^s_l\in\real_{\geq0}$, we let $\VV_L(t^s_l)$ be the set of the mobile leaders that are observed jointly by the agents $\VV$, i.e.,  $\VV_L(t^s_l)=\cup_{i=1}^N\VV_L^i(t^s_l)$  (see Fig. \ref{fig::containment_scenario}). We let $\VV_a(t^s_l)\subset\VV$ be the set of the active agents that observe at least one leader at $t^s_l$, $k\in\mathbb{Z}_{\geq0}$; we assume that $\VV_{\textup{a}}(t^s_l)\neq\emptyset$. In what follows, the objective is to design a distributed control  that enables each follower $i\in\VV$ to derive its local state $\vect{x}^i$ to asymptotically track 
$\Co(\VV_L(t^s_l))$, the convex hull of the set of the location of the observed leaders $\VV_L(t^s_l)$, with a bounded error $e \geq 0$. To simplify notation, we wrote $\Co(\{\vect{x}_{L,j}(t)\}_{j\in\VV_L(t)})$ as   $\Co(\VV_L(t))$.
We state the objective of the containment control as $\|\vect{x}^i(t_k^c)-\bar{\vect{x}}_L(t_k^c)\|\leq e$, $i\in\VV$, 
where $\bar{\vect{x}}_L(t_k^c)\in\Co(\VV_L(t_k^c))$.
The agents~have no knowledge about the motion model of the leaders. Since followers observe the dynamic leaders collaboratively, the tracking error $e$ is expected as the measurement of each active follower needs time to propagate through the network to the rest of the~followers. 

Our solution builds on the key observation that we make below about the convex hull of a set of points $\{\vect{x}_i\}_{i=1}^{m}$ in an Euclidean space.


\begin{lem}
	\label{lem::cov_noverlap}
	Consider a set of points $\{\vect{x}_i\}_{i=1}^{m}$ in $\real^2$ or $\real^3$. Let $\mathcal{S}_j\neq\emptyset$, $j\in\{1,\cdots,s\}$, be a subset of $\{1,\cdots,m\}$. Let $\bar{\vect{x}}_j=\frac{\sum_{k\in\mathcal{S}_j}\vect{x}_k}{|\mathcal{S}_j|}$, $j\in\{1,\cdots,s\}$. Then, the point
$\bar{\vect{x}}=\frac{\sum_{i=1}^s \bar{\vect{x}}_i}{s}$ is a point in $\Co(\{\vect{x}_{j}\}_{j=1}^m)$.
\end{lem}

\begin{proof} \vspace{-0.1in} It is straightforward to confirm that $\bar{\vect{x}}_j\in\Co(\{\vect{x}_{i}\}_{i\in\mathcal{S}_j})$, $j\in\{1,\cdots,s\}$ and $\bar{\vect{x}}\in\Co(\{\bar{\vect{x}}_{i}\}_{i=1}^s)$ (recall the definition of the convex hull).
		Moreover, since $\Co(\{\vect{x}_{k}\}_{k=1}^m)$ is a convex set, we note that $\Co(\{\vect{x}_{i}\}_{i\in\mathcal{S}_j})\subset \Co(\{\vect{x}_{k}\}_{k=1}^m)$, $j\in\{1,\cdots,s\}$. Thus, for $i\in\{1,\cdots,s\}$, $\bar{\vect{x}}_i \in \Co(\{\vect{x}_{j}\}_{j=1}^m)$, and $\Co(\{\bar{\vect{x}}_{i}\}_{i=1}^s)\subset \Co(\{\vect{x}_{i}\}_{i=1}^m)$. As a result, $\bar{\vect{x}}\in\Co(\{\vect{x}_{j}\}_{j=1}^m)$.
\end{proof}

\begin{figure}
	\centering
	\includegraphics[trim= 0 3pt 0 0pt,clip,width=0.28\textwidth]{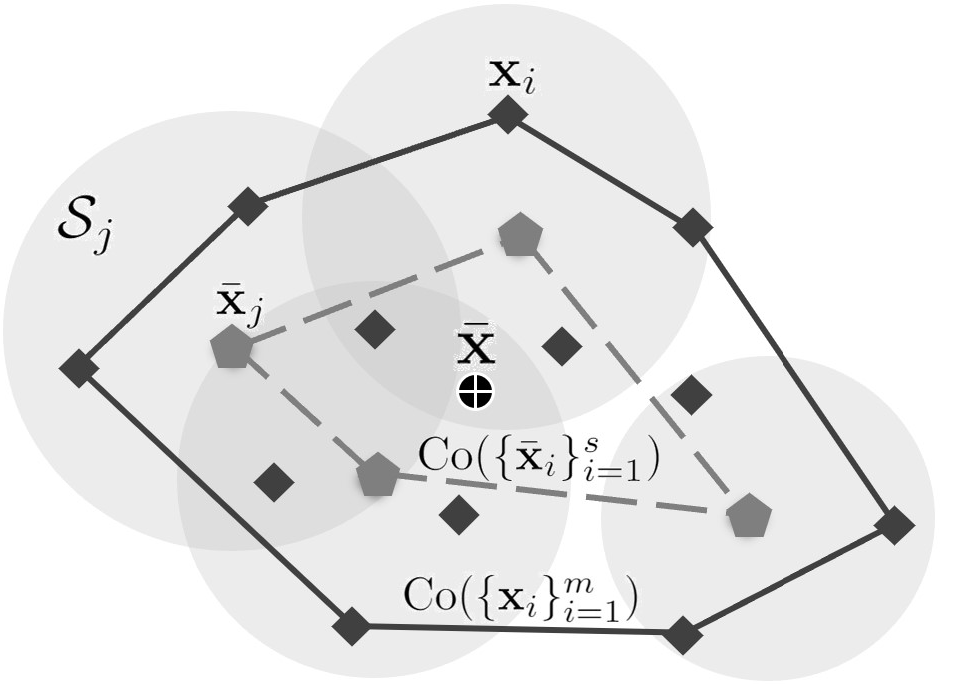}
	\caption{{\small An example graphical demonstration of Lemma~\ref{lem::cov_noverlap}. }}
	\label{fig::lem_convex}
\end{figure}

An example case that demonstrates the result of Lemma~\ref{lem::cov_noverlap} is shown in Fig.~\ref{fig::lem_convex}.
 With the right notation at hand, and the observation made in Lemma~\ref{lem::cov_noverlap}, we are now ready to present in the lemma below  our solution for the containment problem stated above.

\begin{lem}\label{coe::containment_control}
In a containment control problem, let the interaction topology $\GG$ of the followers be a connected graph and suppose that the agents communicate at  $t_k^c=k\delta_c\in\realnonnegative$, $k\in\mathbb{Z}_{\geq0}$. Assume that at each sampling time $t^s_l=l\delta_s \in\realnonnegative$, $l\in\mathbb{Z}_{\geq0}$, we have $\VV_{\textup{a}}(t^s_l)\neq\emptyset$, and the followers are observing the leaders in a zero-order hold fashion, i.e, $\vect{x}_{L,j}(t)=\vect{x}_{L,j}(t^s_l)$, $j\in\VV_L^i(t^s_l)$ and $i\in\VV_{\textup{a}}(t^s_l)$ for $t\in[t^s_l,t^s_{l+1})$. Let
\begin{align}\label{eq:riLs}
\vectsf{r}^i(t^s_l)=\begin{cases}
\frac{\sum_{j\in\VV_L^i(t^s_l)} \vect{x}_{L,j}(t^s_l)}{|\VV_L^i(t^s_l)|},&i\in\VV_{\textup{a}}(t^s_l),\\
~\vect{0},&i\in\VV\backslash\VV_{\textup{a}}(t^s_l).\end{cases}
\end{align} 
Then, $\bar{\vect{x}}_L(t^s_l)=\frac{\sum_{i\in\VV_{\textup{a}}(t^s_l)}\vectsf{r}^i(t^s_l)}{|\VV_{\textup{a}}(t^s_l)|}$ is a point in the convex hull of the leaders $\Co(\VV_L(t^s_l))$. 
Moreover, assume $\|\vect{x}_{L,j}(t^s_{l+1})-\vect{x}_{L,j}(t^s_{l})\|, j\in\{1,\cdots,M\}$, is bounded. If the followers implement active weighted average consensus algorithm \eqref{eq::DT_IF_algorithm} with inputs~\eqref{eq:riLs}, and $\eta^i(t)=1$ if $i\in\VV_{\textup{a}}(t^s_l)$, otherwise, $\eta^i(t)=0$ for $t\in[t^s_l,t^s_{l+1})$, then the tracking error $\|\vect{x}^i(t_k^c)-\overline{\vect{x}}_L(t_k^c)\|$ is~bounded.
\end{lem}
\begin{proof} $\bar{\vect{x}}_L(t^s_l)\in\Co(\VV_L(t^s_l))$ is true by virtue of Lemma~\ref{lem::cov_noverlap}. The boundedness of the tracking error $\|\vect{x}^i(t_k^c)-\overline{\vect{x}}_L(t_k^c)\|$ follows from the guarantees that Theorem~\ref{thm::DT_IF_algorithm} provides.

\end{proof}

Our solution in Lemma~\ref{coe::containment_control} applies to scenarios like in Fig.~\ref{fig::containment_scenario} where the observation sets of the followers have overlap. It is interesting to note that in case of overlapping observations, $\overline{\vect{x}}_L$ is not the centroid of the leaders. Next, note that by virtue of Theorem~\ref{thm::DT_IF_algorithm}, if the leaders are static or move towards a static configuration, the algorithm convergences exactly to $\overline{\vect{x}}_L$. Otherwise, to ensure that the followers stay in the convex hull while tracking $\overline{\vect{x}}_L$ with some error, we may have to require that the convex hull of the leaders should be sufficiently large.

For demonstration, consider a case that $6$ followers with a ring interaction graph aim to follow the convex hull of $10$ leaders in a two dimensional space. The followers observe the leaders at $1$ Hz according to the scenario described below where the set of active followers changes at $t^s_l\!=\!5$ and $t^s_l\!=\!10$ seconds:  
\begin{itemize}\small{
	\item[-] $0\leq t^s_l < 5$ $0\leq t^s_l < 5$: $\VV_L^1(t^s_l)=\{1,4,6,8\}$, $\VV_L^2(t^s_l)=\{2,4,7,8,10\}$, $\VV_L^3(t^s_l)=\{3,4,5,9\}$, $\VV_L^4(t^s_l)=\emptyset$, $\VV_L^5(t^s_l)=\{1,3,9\}$ and $\VV_L^6(t^s_l)=\emptyset$, 
	\item[-] $5\leq t^s_l < 10$: $\VV_L^1(t^s_l)=\{3,5,6,8\}$,  $\VV_L^2(t^s_l)=\{1,2,7,9,10\}$, $\VV_L^3(t^s_l)=\{3,4,5,9\}$, $\VV_L^4(t^s_l)=\emptyset$, $\VV_L^5(t^s_l)=\{1,3,9\}$ and $\VV_L^6(t^s_l)=\{2,5,7,9\}$, 
	\item[-] $10\leq t^s_l \leq 20$: $\VV_L^1(t^s_l)=\{1,2,5,8\}$,  $\VV_L^2(t^s_l)=\{2,3,6,7,10\}$, $\VV_L^3(t^s_l)=\{3,4,5,9\}$, $\VV_L^4(t^s_l)=\{3,10\}$, $\VV_L^5(t^s_l)=\{1,3,9\}$ and $\VV_L^6(t^s_l)=\{2,5,7,9\}$. }
\end{itemize}
 The communication frequency of the followers is $5$ Hz. 
 Figure~\ref{fig::containment_3D} shows 
that the proposed distributed containment control of Lemma~\ref{coe::containment_control} results in a bounded tracking of the convex hull of the observed leaders.
The interested reader can also find an application study of use of our solution in Lemma~\ref{coe::containment_control} in solving containment control for a group of unicycle followers with continuous-time dynamics in our preliminary work~\cite{CY-KS:19}. There, the algorithm in Lemma~\ref{coe::containment_control} is used as an observer to generate the tracking points for the~followers.

\begin{figure}
	\centering
	\includegraphics[trim=0pt  0 0 0,clip,width=0.45\textwidth]{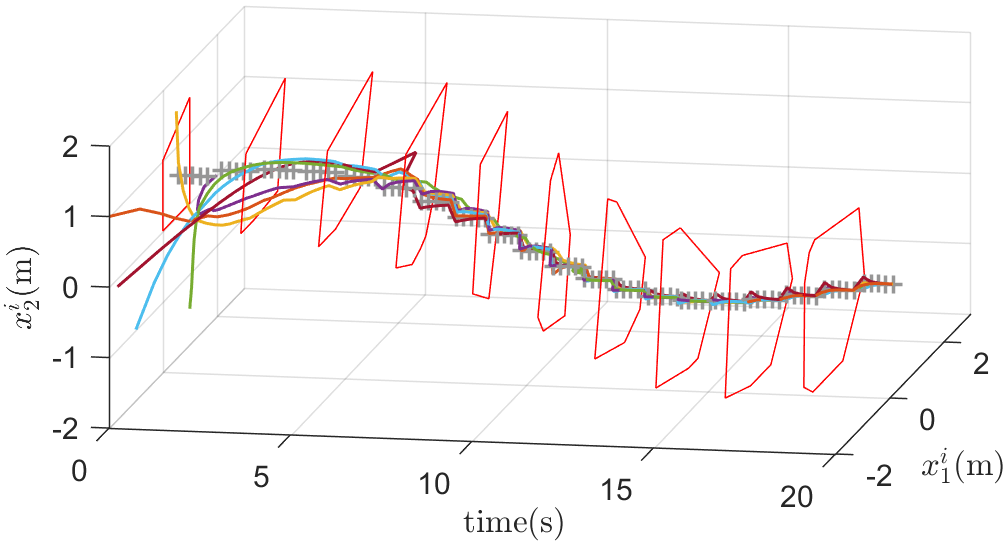}
	\caption{{\small The containment tracking performance of the follower agents while implementing the distributed algorithm~\eqref{eq::DT_IF_algorithm}: the solid curves show the trajectory of $\vect{x}^i$ vs. time, while ``+" show the location of $\bar{\vect{x}}_L(t_k^c)$ of the leaders. The red polygons indicate the convex hull formed by the moving leaders.}}
	\label{fig::containment_3D}
\end{figure}
\section{Conclusion}\label{sec::conclusion}
\vspace{-0.1in}We proposed a dynamic active weighted average consensus algorithm that makes both active and passive agents track the average of the collected reference inputs. The stability and tracking performance were analyzed in both continuous- and discrete-time implementations. We also showed that a containment control can be formulated as an active average consensus problem and solved using our proposed discrete-time algorithm.

\bibliographystyle{ieeetr}%
\vspace{-0.05in}\bibliography{bib/alias,bib/Reference} 

\begin{thebibliography}{10}

\bibitem{ROS-JAF-RMM:07}
R.~Olfati-Saber, J.~A. Fax, and R.~M. Murray, ``Consensus and cooperation in
  networked multi-agent systems,'' {\em Proceedings of the IEEE}, vol.~95,
  no.~1, pp.~215--233, 2007.

\bibitem{SSK-BVS-JC-RAF-KML-SM:19}
S.~S. Kia, B.~V. Scoy, J.~Cort{\'e}s, R.~A. Freeman, K.~M. Lynch, and
  S.~Mart{\'\i}nez, ``Tutorial on dynamic average consensus: The problem, its
  applications, and the algorithms,'' {\em IEEE Control Systems Magazine},
  vol.~39, no.~3, pp.~40--72, 2019.

\bibitem{ROS-RMM:04}
R.~{Olfati-Saber} and R.~M. {Murray}, ``Consensus problems in networks of
  agents with switching topology and time-delays,'' {\em IEEE Transactions on
  Automatic Control}, vol.~49, pp.~1520--1533, Sep. 2004.

\bibitem{CJ::06}
J.~Cort{\'e}s, ``Analysis and design of distributed algorithms for
  x-consensus,'' in {\em Conference on Decision and Control}, pp.~3363--3368,
  IEEE, 2006.

\bibitem{SY::16}
Y.~Shang, ``Finite-time weighted average consensus and generalized consensus
  over a subset,'' {\em IEEE Access}, vol.~4, pp.~2615--2620, 2016.

\bibitem{TY-JDP:14}
T.~{Yucelen} and J.~D. {Peterson}, ``Active-passive networked multiagent
  systems,'' in {\em Conference on Decision and Control}, pp.~6939--6944, 2014.

\bibitem{JDP-TY-GC-SK:15}
J.~D. {Peterson}, T.~{Yucelen}, G.~{Chowdhary}, and S.~{Kannan}, ``Exploitation
  of heterogeneity in distributed sensing: An active-passive networked
  multiagent systems approach,'' in {\em {A}merican {C}ontrol {C}onference},
  pp.~4112--4117, 2015.

\bibitem{JDP-TY-JS-ELP:20}
J.~D. {Peterson}, T.~{Yucelen}, J.~{Sarangapani}, and E.~L. {Pasiliao},
  ``Active-passive dynamic consensus filters with reduced information exchange
  and time-varying agent roles,'' {\em IEEE Transactions on Control Systems
  Technology}, vol.~28, no.~3, pp.~844--856, 2020.

\bibitem{KRP::1998}
R.~P. Kanwal, {\em Distributional Derivatives of Functions with Jump
  Discontinuities}, pp.~99--137.
\newblock Boston, MA: Birkh{\"a}user Boston, 1998.

\bibitem{JM-FT-GE-EM-BA:08}
M.~Ji, G.~Ferrari-Trecate, M.~Egerstedt, and A.~Buffa, ``{Containment Control
  in Mobile Networks},'' {\em IEEE Transactions on Automatic Control}, vol.~53,
  pp.~1972--1975, sep 2008.

\bibitem{WX-LS-SP:14}
X.~Wang, S.~Li, and P.~Shi, ``{Distributed Finite-Time Containment Control for
  Double-Integrator Multiagent Systems},'' {\em IEEE Transactions on
  Cybernetics}, vol.~44, pp.~1518--1528, sep 2014.

\bibitem{HL-GX-LW:12}
H.~Liu, G.~Xie, and L.~Wang, ``Containment of linear multi-agent systems under
  general interaction topologies,'' {\em Systems \& Control Letters}, vol.~61,
  no.~4, pp.~528 -- 534, 2012.

\bibitem{LG-GFT-RC:13}
L.~Galbusera, G.~Ferrari-Trecate, and R.~Scattolini, ``A hybrid model
  predictive control scheme for containment and distributed sensing in
  multi-agent systems,'' {\em Systems \& Control Letters}, vol.~62, no.~5,
  pp.~413 -- 419, 2013.

\bibitem{KZ-SJM-DW:16}
Z.~Kan, J.~M. Shea, and W.~E. Dixon, ``Leader--follower containment control
  over directed random graphs,'' {\em Automatica}, pp.~56--62, 2016.

\bibitem{CY-KS:19}
Y.-F. Chung and S.~S. Kia, ``Distributed dynamic containment control over a
  strongly connected and weight-balanced digraph,'' {\em IFAC-PapersOnLine},
  vol.~52, no.~20, pp.~25--30, 2019.

\bibitem{khalil:02}
H.~K. Khalil, {\em Nonlinear systems}.
\newblock Prentice-Hall, 2002.

\bibitem{ZL-GH:10}
L.~Zhang and H.~Gao, ``Asynchronously switched control of switched linear
  systems with average dwell time,'' {\em Automatica}, vol.~46, no.~5,
  pp.~953--958, 2010.

\bibitem{HJP-MAS:19}
J.~P. Hespanha and A.~S. Morse, ``Stability of switched systems with average
  dwell-time,'' in {\em Conference on Decision and Control}, vol.~3,
  pp.~2655--2660, IEEE, 1999.

\bibitem{HJP::04}
J.~P. Hespanha, ``Uniform stability of switched linear systems: Extensions of
  lasalle's invariance principle,'' {\em IEEE Transactions on Automatic
  Control}, vol.~49, no.~4, pp.~470--482, 2004.

\bibitem{AZ:82}
Z.~Artstein, ``Stability, observability and invariance,'' {\em Journal of
  differential equations}, vol.~44, no.~2, pp.~224--248, 1982.

\bibitem{ZG-HB-YK-MA:02}
G.~Zhai, B.~Hu, K.~Yasuda, and A.~N. Michel, ``Qualitative analysis of
  discrete-time switched systems,'' in {\em Proceedings of the 2002 American
  Control Conference}, vol.~3, pp.~1880--1885, IEEE, 2002.

\bibitem{GJC-CP::06}
J.~C. Geromel and P.~Colaneri, ``Stability and stabilization of discrete time
  switched systems,'' {\em International Journal of Control}, vol.~79, no.~07,
  pp.~719--728, 2006.

\end{thebibliography}

\end{document}